\documentclass[smallextended]{svjour3}    

\usepackage{amssymb, amsmath, amsfonts, mathtools}                                                  
\usepackage{microtype}                                                                              
\usepackage[linesnumbered, ruled, vlined, procnumbered, nokwfunc]{algorithm2e}                      
\usepackage{booktabs, tabularx}                                                                     
\usepackage{environ}                                                                                
\usepackage{hyperref}                                                                               
\usepackage{soul, color}                                                                            
\usepackage[dvipsnames, table]{xcolor}                                                              
\usepackage{pifont}                                                                                 
\usepackage[normalem]{ulem}                                                                         
\usepackage{float}                                                                                  
\usepackage{adjustbox}                                                                              
\usepackage{pdflscape}                                                                              
\usepackage{textgreek}                                                                              

\hypersetup{
  colorlinks = true,
  linktoc = all,
  citecolor = RoyalBlue!95!Black,
  linkcolor = RedOrange!85!Black,
  urlcolor = PineGreen!95!Black,
  filecolor = cyan, 
}

\newcommand*{\prob}[1]{\textsc{#1}}                                                                 
\newcommand*{\probc}[1]{\textsf{#1}}                                                                
\newcommand*{\algo}[1]{\texttt{#1}}                                                                 
\newcommand*{\algoi}[2]{\texttt{#1}(#2)}                                                            
\SetKw{Read}{read}
\SetKw{Write}{write}
\SetKw{Break}{break}
\SetKw{Continue}{continue}
\SetAlgoCaptionSeparator{,}

\DeclarePairedDelimiter{\brnX}{(}{)}
\DeclarePairedDelimiter{\brsX}{[}{]}

\DeclarePairedDelimiter{\brcX}{\{}{\}}
\DeclarePairedDelimiter{\absX}{\lvert}{\rvert}

\DeclarePairedDelimiter{\ceilX}{\lceil}{\rceil}
\DeclarePairedDelimiter{\floorX}{\lfloor}{\rfloor}
\newcommand*{\brn}{\brnX*}                                                                          
\newcommand*{\brs}{\brsX*}                                                                          
\newcommand*{\brc}{\brcX*}                                                                          
\newcommand*{\abs}{\absX*}                                                                          
\newcommand*{\floor}{\floorX*}                                                                      
\newcommand*{\ceil}{\ceilX*}                                                                        

\let\Pr\relax

\DeclareMathOperator{\Pop}{P}

\DeclareMathOperator{\expop}{\mathbb{E}}
\newcommand*{\Pr}[1]{\Pop\brn{#1}}                                                                  
\newcommand*{\Ex}[1]{\expop\brs{#1}}                                                                

\DeclareMathOperator{\ohop}{O}

\newcommand*{\Oh}[1]{\ohop\brn{#1}}                                                                 

\DeclareMathOperator{\Vop}{V}
\DeclareMathOperator{\Eop}{E}
\DeclareMathOperator{\Nbop}{N}
\newcommand*{\V}[1]{\Vop\brn{#1}}                                                                   
\newcommand*{\E}[1]{\Eop\brn{#1}}                                                                   
\newcommand*{\Nb}[1]{\Nbop\brn{#1}}                                                                 
\newcommand*{\Nbr}[2]{\Nbop_{#1}\brn{#2}}                                                           
\let\degd\deg
\let\deg\relax
\newcommand*{\deg}[1]{\degd\brn{#1}}                                                                
\newcommand*{\degr}[2]{\degd_{#1}\brn{#2}}                                                          

\newcommand*{\N}{\mathbb{N}}                                                                        
\newcommand*{\Z}{\mathbb{Z}}                                                                        
\newcommand*{\setb}[2]{\left\{#1 \mid #2\right\}}                                                   

\newcommand*{\dnimply}{\kern.6em\not\kern-.6em \implies}                                            
\newcommand*{\YES}{\texttt{YES}}                                                                    
\newcommand*{\NO}{\texttt{NO}}                                                                      
\newcommand*{\NULL}{\texttt{NULL}}                                                                      

\newcommand*{\sgn}[1]{\sgn\brn{#1}}                                                                 
\newcommand*{\exS}[1]{\textsuperscript{#1}}                                                         
\newcommand*{\tops}[2]{\texorpdfstring{#1}{#2}}

\makeatletter
\newcommand{\ptitle}[1]{\gdef\prob@title{#1}}                                                       
\newcommand{\pobject}[1]{\gdef\prob@object{#1}}                                                     
\newcommand{\pquery}[1]{\gdef\prob@query{#1}}                                                       
\newcommand{\pparam}[1]{\gdef\prob@param{#1}}                                                       

\NewEnviron{problem*}[1][]{%
  \def\prob@type{#1}
  \def\prob@topt{opt}%
  \def\prob@title{}%
  \def\prob@object{}%
  \def\prob@query{}%
  \def\prob@param{}%
  \ifx\prob@type\prob@topt
    \def\prob@qword{Solution}%
  \else
    \def\prob@qword{Question}%
  \fi%
  \BODY
  \ifx\prob@param\@empty
    \def\prob@trow{{\large\prob@title}}
  \else%
    \def\prob@trow{
      \begin{tabular*}{\textwidth}{@{\extracolsep{\fill}}lr}%
        {\large\prob@title} & \textbf{Parameter:}~\prob@param%
      \end{tabular*}%
    }%
  \fi%
  \par\noindent\ignorespaces%
  \ignorespacesafterend%
  \\*[-1ex]%
  \fbox{%
    \hspace{0.5em}%
    \begin{minipage}{0.95\textwidth}%
      \vspace{1.2ex}
      \prob@trow\\[2\fboxsep]
      \textbf{Instance:}~\prob@object\\
      \textbf{\prob@qword:}~\prob@query%
      \vspace{1.4ex}
    \end{minipage}%
    \hspace{0.5em}%
  }%
  \bigskip\par%
}
\makeatother

\title{Approximation in (Poly-) Logarithmic Space\thanks{A preliminary version~\cite{BRS2020MFCS} of this article appeared in the proceedings of MFCS 2020.}
}

\author{Arindam Biswas \and Venkatesh Raman \and Saket Saurabh}
\authorrunning{A.\ Biswas, V.\ Raman and S.\ Saurabh} 
\institute{A.\ Biswas (\Letter) \at The Institute of Mathematical Sciences, HBNI, Chennai, India\\\email{barindam@imsc.res.in} \and V.\ Raman \at The Institute of Mathematical Sciences, HBNI, Chennai, India\\\email{vraman@imsc.res.in} \and S.\ Saurabh \at The Institute of Mathematical Sciences, HBNI, Chennai, India \and University of Bergen, Bergen, Norway\\\email{saket@imsc.res.in}}

\newcommand*{\pDS}{\prob{Dominating Set}}
\newcommand*{\pHS}{\prob{Hitting Set}}
\newcommand*{\pdHS}{d\text{--}\prob{Hitting Set}}
\newcommand*{\pIS}{\prob{Independent Set}}
\newcommand*{\pVC}{\prob{Vertex Cover}}
\newcommand*{\pSC}{\prob{Set Cover}}
\newcommand*{\pTVD}{\prob{Triangle-Free Deletion}}
\newcommand*{\pTFVS}{\prob{Tournament FVS}}
\newcommand*{\pCVD}{\prob{Cluster Deletion}}
\newcommand*{\pSVD}{\prob{Split Deletion}}
\newcommand*{\pThVD}{\prob{Threshold Deletion}}
\newcommand*{\pCoVD}{\prob{Cograph Deletion}}
\newcommand*{\pSort}{\prob{Sorting}}
\newcommand*{\pSelect}{\prob{Selection}}
\newcommand*{\pLP}{\prob{Linear Programming}}
\newcommand*{\pUSTCON}{\prob{USTCON}}
\newcommand*{\pSTCON}{\prob{STCON}}

\newcommand*{\cP}{\probc{P}}

\newcommand*{\cL}{\probc{L}}
\newcommand*{\cNL}{\probc{NL}}

\RequirePackage{fix-cm}
\usepackage{graphicx}

\usepackage{amsthm}
\usepackage[misc]{ifsym}
\smartqed  
\begin{document}
\maketitle
\begin{abstract}
We develop new approximation algorithms for classical graph and set problems in the RAM model under space constraints. As one of our main results, we devise an algorithm for $\pdHS{}$ that runs in time $n^{\Oh{d^2 + (d / \epsilon)}}$, uses $\Oh{(d^2 + (d / \epsilon))\log {n}}$ bits of space, and achieves an approximation ratio of $\Oh{(d / \epsilon) n^{\epsilon}}$ for any positive $\epsilon \leq 1$ and any $d \in \N$. In particular, this yields a factor-$\Oh{\log{n}}$ approximation algorithm which runs in time $n^{\Oh{\log{n}}}$ and uses $\Oh{\log^2{n}}$ bits of space (for constant $d$). As a corollary, we obtain similar bounds for $\pVC{}$ and several graph deletion problems.

For bounded-multiplicity problem instances, one can do better. We devise a factor-$2$ approximation algorithm for $\pVC{}$ on graphs with maximum degree $\Delta$, and an algorithm for computing maximal independent sets, both of which run in time $n^{\Oh{\Delta}}$ and use $\Oh{\Delta \log{n}}$ bits of space. For the more general $\pdHS{}$ problem, we devise a factor-$d$ approximation algorithm which runs in time $n^{\Oh{d{\delta}^2}}$ and uses $\Oh{d {\delta}^2 \log{n}}$ bits of space on set families where each element appears in at most $\delta$ sets.

For $\pIS{}$ restricted to graphs with average degree $d$, we give a factor-$(2d)$ approximation algorithm which runs in polynomial time and uses $\Oh{\log{n}}$ bits of space. We also devise a factor-$\Oh{d^2}$ approximation algorithm for $\pDS{}$ on $d$-degenerate graphs which runs in time $n^{\Oh{\log{n}}}$ and uses $\Oh{\log^2{n}}$ bits of space. For $d$-regular graphs, we show how a known randomized factor-$\Oh{\log{d}}$ approximation algorithm can be derandomized to run in time $n^{\Oh{1}}$ and use $\Oh{\log n}$ bits of space.

Our results use a combination of ideas from the theory of kernelization, distributed algorithms and randomized algorithms.

\keywords{approximation \and logspace \and logarithmic \and log \and space \and small \and limited \and memory \and ROM \and read-only}
\end{abstract}

\section{Introduction}
This paper examines the classical approximation problems $\pVC{}$, $\pHS{}$ and $\pDS{}$ in the RAM model under additional polylogarithmic space constraints. We devise approximation algorithms for these problems which use polylogarithmic space in general and $\Oh{\log{n}}$ bits of space on certain special input types.

In the absence of space constraints, the greedy heuristic is a good starting point for many approximation algorithms. For $\pSC{}$, it even yields optimal (under certain complexity-theoretic assumptions) approximation ratios~\cite{AMS2006TALG,DS2014STOC}. However, the heuristic inherently changes the input in some way. In a space-constrained setting, this is asking for too much: the input is immutable, and the amount of auxiliary space available (in our case, polylogarithmic in the input size) is not sufficient to register changes to the input.

Linear programming is another tool that plays a central role in the design of approximation algorithms. While it yields competitive approximations in polynomial time when space is not constrained, it is known that under logarithmic-space reductions, it is $\cP{}$-complete to approximate $\pLP{}$ to any constant factor~\cite{Ser1991IPL}. Such a result even holds for $\pLP{}$ instances where all coefficients are positive~\cite{TX1998ParallelProcessLett}.

\paragraph*{Machine Model.}\ We use the standard RAM model with an additional polylogarithmic space constraint. For inputs $n$ bits in length, memory is organized as words of length $\Oh{\log{n}}$, which allows any input element to be addressed using a single word of memory. Integer arithmetic operations on pairs of words and single-word memory access operations take constant time. This is referred to in the literature as the \emph{word} RAM model.

The input (a graph or family of sets) is provided to the algorithm using some canonical encoding, which can be read but not modified, i.e.\ the algorithm has read-only access to the input. The algorithm uses some auxiliary memory, to which it has read-write access, and in the setting of this paper, the amount of such memory available is bounded by a polynomial in $\log{n}$. Output is written to a stream: once something is output, the algorithm cannot read it back later on as it executes. We count the amount of auxiliary memory used in units of $1$ bit, and the objective is to use as little auxiliary memory as possible.




\subsection*{Our Results}
\paragraph*{$\pdHS{}$ and Vertex Deletion Problems.}\ An instance of the\\$\pdHS{}$ problem consists of a ground set $U$ and a family $\mathcal{F}$ of subsets of $U$ of size at most $d \in \N$, and the objective is to find a subset of the ground set that intersects every set in the family.
\begin{itemize}
    \item We develop a factor-$\Oh{(d / \epsilon) n^{\epsilon}}$ approximation algorithm for $\pdHS{}$ which runs in time $n^{\Oh{d^2 + (d/\epsilon)}}$ and uses $\Oh{(d^2 + (d / \epsilon))\log{n}}$ bits of space (Section~\ref{sect:hs_vc_del}), where $\epsilon \leq 1$ is an arbitrary positive number and $d$ is a fixed positive integer. In particular, this yields a factor-$\Oh{d \log n}$ approximation algorithm for the problem which uses $\Oh{\log^2{n}}$ bits of space. As an application, we show how the algorithm can be used to approximate various \emph{deletion} problems with similar space bounds. From this, we derive a factor-$\Oh{(1 / \epsilon) n^{\epsilon}}$ (for arbitrary positive $\epsilon \leq 1$) approximation algorithm for $\pVC{}$ that runs in time $n^{\Oh{1 /\epsilon)}}$ and uses $\Oh{(1 / \epsilon) \log{n}}$ bits of space.

    \item We give a simple factor-$2$ approximation algorithm for $\pVC{}$ on graphs with maximum degree $\Delta$ which runs in time $n^{\Oh{\Delta}}$ and uses $\Oh{\Delta \log{n}}$ bits of space (Section~\ref{ssct:bd_vc_2_approx}).

    \item For $\pdHS{}$ instances where each element appears in at most $\delta$ sets, we devise a factor-$d$ approximation algorithm (Section~\ref{ssct:bd_dhs_d_approx}), generalizing the above result. The algorithm runs in time $n^{\Oh{d {\delta}^2}}$ and uses $\Oh{d {\delta}^2 \log{n}}$ bits of space.
\end{itemize}

\paragraph*{$\pDS{}$.}\ In the $\pDS{}$ problem, the objective is to find a vertex set of minimum size in a graph such that all other vertices are adjacent to some vertex in the set.
\begin{itemize}
    \item We give a factor-$\Oh{\sqrt{n}}$ approximation algorithm for graphs excluding $C_4$ (a cycle on $4$ vertices) as a subgraph, which runs in polynomial time and uses $\Oh{\log{n}}$ bits of space (Section~\ref{ssct:c4_ds_approx}).

    \item Graphs of bounded degeneracy form a large class which includes planar graphs, graphs of bounded genus, graphs excluding a fixed graph $H$ as a (topological) minor and graphs of bounded expansion. For graphs with degeneracy $d$, we give a factor-$\Oh{d^2}$ approximation algorithm which uses $\Oh{\log^2{n}}$ bits of space. (Section~\ref{ssct:dgn_ds_approx}).

    \item Additionally, for graphs in which each vertex has degree $d$, i.e.\ $d$-regular graphs, we exhibit a factor-$\Oh{\log{d}}$ approximation algorithm for\\$\pDS{}$ (Section~\ref{ssct:reg_ds_approx}) which is an adaptation of known results to the constrained-space setting.
\end{itemize}

\paragraph*{$\pIS{}$.}\ An instance of the $\pIS{}$ problem consists of a graph, and the objective is to find an \emph{independent set} of maximum size i.e.\ a set of vertices with no edges between them.
\begin{itemize}
    \item We show how a known factor-$(2d)$ approximation algorithm for\\$\pIS$ on graphs with average degree $d$ can be implemented to run in polynomial time and use $\Oh{\log n}$ bits of space (Section~\ref{ssct:avd_is_approx}).

    \item For the related problem of finding maximal independent sets, we devise an algorithm which runs in time $n^{\Oh{\Delta}}$ and uses $\Oh{\Delta \log{n}}$ bits of space (Theorem~\ref{thrm:bd_maxl_is}) on graphs with maximum degree $\Delta$.
\end{itemize}

\paragraph*{Remark.}\ In various statements, we use the verb \emph{enumerate} to explicitly indicate that the structures being computed are produced in serial fashion, and read by other procedures later on or output as a stream.

\subsection*{Related Work}
Small-space models such as the streaming model and the in-place model have been the subject of much research over the last two decades (see~\cite{McG2014SIGMODRec,CMR2018TALG,CMRS2018ESA} and references therein). In the streaming model, in addition to the space constraint, the algorithm is also required to read the input in a specific (possibly adversarial) sequence in one or more passes. The in-place model, on the other hand, allows the memory used for storing the input to be modified. The read-only RAM model we use is distinct from both these models. 

Historically, the read-only model has been studied from the perspective of time–space tradeoff lower bounds, particularly for problems like $\pSort{}$~\cite{BC1982SICOMP,BFK+1981JCSS,Bea1991SICOMP,PR1998FOCS,PP2002SODA} and $\pSelect{}$~\cite{MP1980TCS,Fre1987JCSS,MR1996TCS,RR1999NordJComput}.

The earliest graph problems studied in this model were the undirected and directed graph reachability problems (resp.\ $\pUSTCON$ and $\pSTCON$) in connection with the complexity classes $\cL{}$ and $\cNL$. Savitch~\cite{Sav1970JCSS} showed that on input graphs with $n$ vertices, $\pSTCON$ (and therefore also $\pUSTCON$) can be solved using $\Oh{\log^2{n}}$ bits of space. This bound was gradually whittled down over more than two decades, and eventually Reingold~\cite{Rei2008JACM} showed that that $\pUSTCON{}$ can be solved using $\Oh{\log{n}}$ bits of space. Graph recognition problems for classes such as bipartite~\cite{Rei1984JACM} and planar~\cite{AM2004InfComput} graphs have been shown to reduce in logarithmic space to $\pUSTCON{}$, which implies logarithmic-space algorithms for these problems as well. The more general problem of recognizing bounded-genus graphs is also known to be solvable in logarithmic space, by an algorithm of Elberfeld and Kawarabayashi~\cite{EK2014STOC}.

For the task of enumerating BFS and DFS traversal sequences, a number of algorithms have been developed~\cite{EHK2015STACS,AIK+2014ISAAC,CRS2017JCSS,Hag2020Algorithmica} which run in polynomial (many of them close to linear) time and their space usage in bits is linear in the number of vertices or edges. The read-only RAM model has also been studied in relation to polynomial-time-solvable search problems~\cite{Yam2013COCOA} and the approximation properties of search problems that can be solved in nondeterministic logarithmic space~\cite{Tan2007TheoryComputSyst}.


Algorithms for the PRAM model---where a number of processors run in parallel and access a common area of memory to solve problems---can sometimes be translated into sequential algorithms that use small space. A known reduction~\cite{Pap1994book} allows one to convert any PRAM algorithm with parallel running time $s(n)$ to a sequential algorithm that uses $s(n) \log^c{n}$ ($c = 1$ or $2$, depending on the PRAM variant) bits of space. These algorithms, however, do not necessarily have polynomial running times.
The PRAM algorithm of Luby~\cite{Lub1986SICOMP} for finding maximal independent sets in a graph can be used to $2$-approximate $\pVC$ (better approximation ratios are unlikely~\cite{KR2008JCSS}). Implemented in the sequential RAM model, it uses $\Oh{\log^2{n}}$ bits of space. For the more general problem of finding maximal independent sets in hypergraphs, the recent PRAM algorithm of Harris~\cite{Har2019TALG} yields a polylogarithmic-space algorithm for hypergraphs with edges of fixed size.

Our scheme for $\pdHS{}$ trades approximation factor against space to yield a family of algorithms that use $\Oh{(d^2 + (d / \epsilon))\log n)}$ bits of space and produce $\Oh{(d / \epsilon) n^{\epsilon}}$-approximate solutions for any positive $\epsilon \leq 1$. As a corollary, we obtain an $\Oh{d \log{n}}$-approximation algorithm that uses $\Oh{\log^2 n}$ bits of space. On graphs with maximum degree $\Delta$, our approximation algorithm for $\pVC{}$ uses $\Oh{\Delta \log{n}}$ bits of space and produces $2$-approximate solutions.

Berger et al.~\cite{BRS1994JCSS} gave a PRAM algorithm for $\pSC$ which can be implemented in the sequential RAM model to $\Oh{\log{n}}$-approximate $\pDS$ in $\Oh{\log^4{n}}$ bits of space. See also~\cite{Tre1998Algorithmica,LN1993STOC}, which give parallel approximation algorithms for $\pLP{}$, and see~\cite{LO2017SPAA}, which gives tight approximation ratios for CSP's using semi-definite programming in the PRAM model. Our algorithms for $\pDS{}$ are simpler and more direct, and work for a large class of graphs while using $\Oh{\log^2{n}}$ bits of space.

\subsection*{Our Techniques}
As noted earlier, the greedy heuristic causes changes to the input, which our model does not permit. To get around this, we use a \emph{staggered} greedy approach in which the solution is constructed in a sequence of greedy steps to approximate $\pVC$ and $\pdHS{}$ on bounded-multiplicity instances (Section~\ref{sect:layering}). By combining this with data reduction rules from kernelization algorithms, we also obtain approximations for $\pVC{}$ and more generally $\pdHS{}$ (Section~\ref{sect:hs_vc_del}), and restricted versions of $\pDS$ (Sections~\ref{ssct:c4_ds_approx}~and~\ref{ssct:dgn_ds_approx}). In Section~\ref{sect:randomization}, we use $2$-universal hash families constructible in logarithmic space to derandomize certain randomized sampling procedures for approximating $\pIS$ on graphs of bounded average degree and $\pDS$ on regular graphs.

\section{Preliminaries}
\paragraph*{Notation.}\ $\N$ denotes the set of natural numbers $\brc{0, 1, \dotsc}$ and $\Z^+$ denotes the set of positive integers $\brc{1, 2, \dotsc}$. For $n \in \Z^+$, $[n]$ denotes the set $\brc{1, 2, \dotsc, n}$. Let $G$ be a graph. Its vertex set is denoted by $\V{G}$, and its edge set by $\E{G}$. The degree of a vertex $v$ is denoted by $\deg{v}$, and for a set $S \subseteq \V{G}$ or a subgraph $H$ of $G$, $\degr{S}{v}$ denotes the degree of $v$ in $G[S]$ and $\degr{H}{v}$ denotes the degree of $v$ in $H$.

\paragraph*{Known Results.}\ The following result combines a known logarithmic-space implementation of the Buss kernelization rule~\cite{BG1993SICOMP} for $\pVC{}$ with the observation that the kernel produced is itself a vertex cover.

\begin{proposition}[Cai et al.~\cite{CCDF1997AnnPureApplLogic}, Theorem 2.3]\label{prop:ceal_vc_kernel}
	There is an algorithm which takes as input a graph $G$ and $k \in \N$, and either determines that $G$ has no vertex cover of size at most $k$ or enumerates a vertex cover for $G$ with at most $2k^2$ edges. The algorithm runs in time $\Oh{kn}$ and uses $\Oh{\log{n}}$ bits of space. 
\end{proposition}

$\pVC{}$ is a special case of $\pdHS{}$ ($d \in \N$, a constant), an instance of which comprises a family $\mathcal{F}$ of subsets of a ground set which all have size at most $d$. The objective is to compute a minimum \emph{hitting set} for $\mathcal{F}$, i.e.\ a subset of the ground set which intersects each set in $\mathcal{F}$. The next proposition shows that a result similar to the one above also holds for this generalization.
\begin{proposition}[Fafianie and Kratsch~\cite{FK2015MFCS}, Theorem 1]\label{prop:fk_hs_kernel}
	There is an algorithm which takes as input a family $\mathcal{F}$ of $d$-subsets ($d \in \N$, a constant) of a ground set $U$ and $k \in \N$, and either determines that $\mathcal{F}$ has no hitting set of size at most $k$ or enumerates a subfamily $\mathcal{F}'$ of $\mathcal{F}$ with $\Oh{{(k + 1)}^d}$ sets which is equivalent to it: $\mathcal{F}$ has a hitting set of size at most $k$ if and only if $\mathcal{F}'$ has a hitting set of size at most $k$. The algorithm runs in time $n^{\Oh{d^2}}$ and uses $\Oh{d^2 \log n}$ bits of space.
\end{proposition}

\subsection{Solving $\pVC{}$ and $\pIS{}$ on trees}
When the input is a tree on $n$ vertices, one can enumerate \emph{exact} solutions for $\pVC{}$ and $\pIS{}$ in time $n^{\Oh{1}}$ and $\Oh{\log{n}}$ bits of space, as we show below.

The decision versions of both problems can be expressed in monadic second-order logic (MSOL), allowing them to be solved using Courcelle's theorem~\cite{Cou1990InfComput}, and a result of Elberfeld et al.~\cite{EJT2010FOCS} shows that this can be done in polynomial time and logarithmic space on graphs of constant treewidth. Consequently, there are polynomial-time logarithmic-space algorithms for the decision versions of $\pVC{}$ and $\pIS{}$ on trees. Using pre-existing algoritihms, the search versions of these problems can be solved in polynomial time using $\Oh{\log^2{n}}$ bits of space~\cite{BCR+2015COCOON}. For completeness, we give here elementary algorithms for the special case of trees.

\begin{lemma}\label{lemm:tree_vc}
	Given a tree $T$ on $n$ vertices, \hyperref[algo:TreeVtxCover]{TreeVtxCover} enumerates a minimum vertex cover for $T$ in time $n^{\Oh{1}}$ using $\Oh{\log{n}}$ bits of space.
\end{lemma}

In any graph, the complement of a vertex cover is an independent set, so if \hyperref[algo:TreeVtxCover]{TreeVtxCover} outputs a minimum vertex cover $C$ for a tree $T$, its complement $\V{T} \setminus C$ is a maximum independent set in $T$. Combining this with the fact that given oracles for $\V{T}$ and $C$, the complement can be enumerated in time $n^{\Oh{1}}$ using $\Oh{\log{n}}$ bits of additional space, we have the following corollary.

\begin{corollary}\label{corr:tree_is}
	In a tree of order $n$, one can enumerate a maximum independent set in time $n^{\Oh{1}}$ using $\Oh{\log{n}}$ bits of space.
\end{corollary}

\begin{algorithm}
\KwIn{$T = (V, E)$, a tree}
\KwOut{a mininmum vertex cover for $T$}

	let $r$ be an arbitrary vertex of $T$\;
	\ForEach{$v \in \V{T}$}{
		\If{\algoi{IsInVC}{$v, r, T$}}{
			output $v$\;
		}
	}
    \SetKwFunction{vcproc}{IsInVC}
    \SetKwProg{proc}{Procedure}{}{}
    \proc{\vcproc{$v, r, T$}\tcp*[f]{$T$ a tree rooted at $r$, $v$ a vertex in $\V{T}$}}{\label{proc:TreeVtxCover.IsInVC}
		generate a post-order traversal $L$ for $T$ with $r$ as the root\;
		seek $L$ to the first leaf in the subtree of $T$ rooted at $v$\;
		$visited\_vertex \gets \NULL{}$\;
		$visited\_included \gets \NO{}$\;
		\ForEach{$u \in L$}{
			\eIf{$u$ is a leaf}{
				$visited\_vertex \gets u$\;
				$visited\_included \gets \NO{}$\;
			}(\tcp*[f]{$u$ is not a leaf; $u$ is the parent of $visited\_vertex$}){
				$visited\_vertex \gets u$\;
				\eIf{not $visited\_included$}{
					\If{$u = v$}{
						\Return{\YES{}}\;
					}
					$visited\_included \gets \YES{}$\tcp*{include $u$}
				}(\tcp*[f]{last-visited vertex was included}){
					\If{$u = v$}{
						\Return{\NO{}}\;
					}
					$visited\_included \gets \NO{}$\tcp*{do not include $u$}
					seek $L$ to $u$'s parent\tcp*{vertices in subtrees of $u$'s unvisited siblings can be ignored}
				}
			}
		}
    }

\caption{TreeVtxCover: enumerate a minimum vertex cover}\label{algo:TreeVtxCover}
\end{algorithm}

We now prove Lemma~\ref{lemm:tree_vc}. \hyperref[algo:TreeVtxCover]{TreeVtxCover} operates by rooting $T$ at an arbitrary vertex $r \in \V{T}$ and enumerates a vertex cover $S$ obtained by repeatedly applying the following rule.

\begin{quote}\label{quot:rule_vct} 
	\textsf{Rule VCT} \emph{Include the the parents of leaves on the bottom level of $T$ in $S$, then delete from $T$ the included vertices, their children, and all edges incident with them.}
\end{quote}

The fact that $S$ is a minimum vertex cover follows directly from the observation that to cover the edges of $T$ incident with the leaves at the bottom level, picking the parents of those leaves is at least as good as any other choice of covering vertices.

\begin{proof}[Proof of Lemma~\ref{lemm:tree_vc}]
Observe that at any intermediate stage in the repeated application of Rule \hyperref[quot:rule_vct]{VCT}, a vertex is a leaf on the bottom level of $T$ if a previous application of the rule deleted all of its children, i.e.\ all of them were included in $S$. Thus, any vertex $v \in \V{T}$, is in $S$ if and only if $S$ does not contain all of its children. 

Instead of repeatedly deleting vertices from $T$, Procedure~\hyperref[proc:TreeVtxCover.IsInVC]{IsInVC} in the algorithm determines membership in $S$ by performing what is essentially a post-order traversal of $T$. In the post-order traversal, to determine if a vertex $v$ is in $S$, the only information necessary is whether at least one of $v$'s children is not in $S$, which the procedure stores in the variable $visited\_included$. If such a child vertex is encountered, the procedure determines that $v$ is in $S$, and skips the rest of the subtree rooted at $v$. 

The post-order traversal used by the procedure can be generated from a DFS traversal of $T$, which can be enumerated using $\Oh{\log{n}}$ bits of space~\cite{CM1987JAlgorithms}. The constantly-many variables appearing in the algorithm also use $\Oh{\log{n}}$ bits of space total. Therefore, the overall space usage of the algorithm is $\Oh{\log{n}}$ bits and it runs in time $n^{\Oh{1}}$.
\end{proof}

\subsection{Presenting modified structures using oracles}\label{ssct:graph_oracles}
Our algorithms repeatedly ``delete'' vertices (or elements), but as they only have read-only access to the graph (or set family), we require a way to implement these deletions using a small amount of auxiliary space. Towards that, we prove the following theorem.

\begin{theorem}\label{thrm:graph_oracles}
	Let $G = G_0 = (V, E)$ be a graph on $n$ vertices and let $G_i$ ($i \in [k]$) be obtained from $G_{i - 1}$ by deleting a set $S_i \subseteq \V{G_{i - 1}}$ consisting of all vertices $v \in \V{G_{i - 1}}$ which satisfy a property that can be checked (given access to $G_{i - 1}$) using $\Oh{\log{n}}$ bits of space.

	Given read-only access to $G$, one can, for each $i \in [k]$, enumerate and answer membership queries for $S_i$, $V_i = \V{G_i}$ and $E_i = \E{G_i}$ in time $n^{\Oh{i}}$ using $\Oh{i \log{n}}$ bits of space.
\end{theorem}
\begin{proof}
	For each $i \in [k]$ let $\algo{Check}_i(G_{i - 1}, v)$ be the algorithmic check which, given (oracle) access to $G_{i - 1}$, determines whether $v \in V_{i - 1}$ satisfies the condition for inclusion in $S_i$.

	To provide oracle access to $G_i, V_i$ and $E_i$, it suffices to compute, for $v \in V$ and $uw \in E$, the predicates $\brs{v \in V_i}$ and $\brs{uw \in E_i}$. A vertex is in $V_i$ if and only if it is in $V_{i - 1}$ and it is not in $S_i$. Similarly, an edge is in $E_i$ if and only if it is in $E_{i - 1}$ and neither of its endpoints are in $S_i$. Thus, we have the following relations.
	\begin{align}
		\brs{v \in V_i} &\equiv \brs{v \in V_{i - 1}} \wedge \neg \algo{Check}_i(G_{i - 1}, v)\label{lemm:oracle_rel_1}\\
		\brs{uw \in E_i} &\equiv \brs{uw \in E_{i - 1}} \wedge \neg (\algo{Check}_i(G_{i - 1}, u) \vee \algo{Check}_i(G_{i - 1}, w))\label{lemm:oracle_rel_2}
	\end{align}

	To compute each of these predicates for $G_i$, we require oracle access to $G_{i - 1}$, which in turn involves computing the predicates $\brs{v \in V_{i - 1}}$ and $\brs{uw \in E_{i - 1}}$. Since $\algo{Check}_i(G_{i - 1}, v)$ uses $\Oh{\log{n}}$ bits of space, the number of operations needed to compute it is at most $n^{\Oh{1}}$.
	
	Let $p_i$ (resp.\ $q_i$) be the amount of space used to compute the predicate $\brs{v \in V_i}$ (resp.\ $\brs{uw \in E_i}$), and let $s_i$ (resp.\ $t_i$) be the time needed to compute the predicate $\brs{v \in V_i}$ (resp.\ $\brs{uw \in E_i}$). From Relations~\ref{lemm:oracle_rel_1}~and~\ref{lemm:oracle_rel_2} and the fact that $\algo{Check}_i$ accesses $G_{i - 1}$ at most $n^{\Oh{1}}$ times, we see that these quantities satisfy the following relations.
	\begin{align}
		&p_i = p_{i - 1} + \Oh{\log{n}},\ &q_i = q_{i - 1} + \Oh{\log{n}}\\
		&s_i = s_{i - 1} + n^{\Oh{1}} (s_{i - 1} + t_{i - 1}),\ &t_i = t_{i - 1} + n^{\Oh{1}} (s_{i - 1} + t_{i - 1})
	\end{align}
	
	It is easy to see that these recurrences solve to $p_i, q_i = \Oh{i \log{n}}$ and $s_i, t_i = n^{\Oh{i}}$, so both predicates can be computed in time $n^{\Oh{i}}$ using $\Oh{i \log{n}}$ bits of space.
	
	With oracle access to $G_{i - 1}$, the predicate $\brs{v \in S_i}$ can be computed simply as $\algo{Check}_i(G_{i - 1}, v)$, from which enumerating $V_i$ (resp. $E_i$ and $S_i$) is straightforward: enumerate $V$ (resp. $E$ and $V$) and suppress vertices $v$ (resp. edges $uw$ and vertices $z$) which fail the predicate $\brs{v \in V_i}$ (resp. $\brs{uw \in E_i}$ and $\brs{z \in S_i}$).
	
	As the most space-hungry operations are the membership queries, the enumeration can also be performed using $\Oh{i \log{n}}$ bits of space. The enumeration needs time $n^{\Oh{i}}$ for each element of $V$ and $E$, and since $\abs{V}, \abs{E} = \Oh{n^2}$, the total time needed is also $n^{\Oh{i}}$.
\end{proof}

The above result also generalizes to set families if we assume that at each stage where elements are deleted, the algorithmic check which determines the elements to be deleted uses $\Oh{c \log{n}}$ bits of space ($c > 0$).

\begin{theorem}\label{thrm:hypergraph_oracles}
	Let $U$ be a ground set with $n$ elements and $\mathcal{F} = \mathcal{F}_0$ be a family of subsets of $U$, each of size at most $d$ ($d \in \N$, a constant). Let $\mathcal{F}_i$ ($i \in [k]$) be a subfamily of $\mathcal{F}_{i - 1}$ obtained by deleting all sets in $\mathcal{F}_{i - 1}$ which intersect some set $S_i \subseteq U_{i - 1}$ and let $U_i = U_{i - 1} \setminus S_i$. Suppose additionally that given access to $U_{i - 1}$ and $\mathcal{F}_{i - 1}$, $S_i$ can be determined using $\Oh{c \log{n}}$ bits of space.

	Given read-only access to $\mathcal{F}$, one can, for each $i \in [k]$, enumerate and answer membership queries for $S_i$, $U_i$ and $\mathcal{F}_i$ in time $n^{\Oh{ic}}$ using $\Oh{ic \log{n}}$ bits of space.
\end{theorem}
\begin{proof}
	The proof is essentially the same as that of Theorem~\ref{thrm:graph_oracles}. For each $i \in [k]$ let $\algo{Check}_i(U_{i - 1}, \mathcal{F}_{i - 1}, e)$ be the algorithmic check which, given (oracle) access to $U_{i - 1}$ and $\mathcal{F}_{i - 1}$, determines whether $e \in U_{i - 1}$ satisfies the condition for inclusion in $S_i$.

	As in the proof of Theorem~\ref{thrm:graph_oracles}, it suffices to compute, for $e \in U$ and $A \in \mathcal{F}$, the predicates $\brs{e \in U_i}$ and $\brs{A \in \mathcal{F}_i}$. An element is in $U_i$ if and only if it is in $U_{i - 1}$ and it is not in $S_i$. Similarly, a set is in $\mathcal{F}_i$ if and only if it is in $\mathcal{F}_{i - 1}$ and it does not intersect $S_i$. Thus, we have the following relations.
	\begin{align}
		\brs{e \in U_i} &\equiv \brs{e \in U_{i - 1}} \wedge \neg \algo{Check}_i(U_{i - 1}, \mathcal{F}_{i - 1}, e)\label{lemm:hyp_oracle_rel_1}\\
		\brs{A \in \mathcal{F}_i} &\equiv \brs{A \in \mathcal{F}_{i - 1}} \wedge \neg \brn{\bigvee_{e \in A} \algo{Check}_i(U_{i - 1}, \mathcal{F}_{i - 1}, e)}\label{lemm:hyp_oracle_rel_2}
	\end{align}

	To compute these predicates for $U_i$ and $\mathcal{F}_i$, we require oracle access to $U_{i - 1}$ and $\mathcal{F}_{i - 1}$, which in turn involves computing the predicates $\brs{u \in U_{i - 1}}$ and $\brs{A \in \mathcal{F}_{i - 1}}$. Since $\algo{Check}_i(U_{i - 1}, \mathcal{F}_{i - 1}, e)$ uses $\Oh{c \log{n}}$ bits of space, the number of operations needed to compute it is at most $n^{\Oh{c \log{n}}}$.
	
	Let $p_i$ (resp.\ $q_i$) be the amount of space used to compute the predicate $\brs{e \in U_i}$ (resp.\ $\brs{A \in \mathcal{F}_i}$), and let $s_i$ (resp.\ $t_i$) be the time needed to compute the predicate $\brs{e \in u_i}$ (resp.\ $\brs{A \in \mathcal{F}_i}$). From Relations~\ref{lemm:hyp_oracle_rel_1}~and~\ref{lemm:hyp_oracle_rel_2} and the fact that $\algo{Check}_i$ accesses $\mathcal{F}_{i - 1}$ at most $n^{\Oh{c \log{n}}}$ times, we see that these quantities satisfy the following relations.
	\begin{align}
		&p_i = p_{i - 1} + \Oh{c\log{n}},\ &q_i = q_{i - 1} + \Oh{c\log{n}}\\
		&s_i = s_{i - 1} + n^{\Oh{c \log{n}}} (s_{i - 1} + t_{i - 1}),\ &t_i = t_{i - 1} + n^{\Oh{c \log{n}}} (s_{i - 1} + t_{i - 1})
	\end{align}
	
	It is easy to see that these recurrences solve to $p_i, q_i = \Oh{ic \log{n}}$ and $s_i, t_i = n^{\Oh{ic}}$, so both predicates can be computed in time $n^{\Oh{ic}}$ using $\Oh{ic \log{n}}$ bits of space.
	
	With oracle access to $U_{i - 1}$, the predicate $\brs{e \in S_i}$ can be computed as $\algo{Check}_i(U_{i - 1}, \mathcal{F}_{i - 1}, e)$, from which enumerating $U_i$ (resp. $\mathcal{F}_i$ and $S_i$) is straightforward: enumerate $U$ (resp. $\mathcal{F}$ and $U$) and suppress elements $e$ (resp. sets $A$ and elements $f$) which fail the predicate $\brs{e \in U_i}$ (resp. $\brs{A \in \mathcal{F}_i}$ and $\brs{f \in S_i}$).

	The most space-hungry operations are the membership queries, so the enumeration can also be performed using $\Oh{ic \log{n}}$ bits of space. The enumeration needs time $n^{\Oh{ic}}$ for each element of $U$ and $\mathcal{F}$. Since $\abs{U}, \abs{\mathcal{F}} = \Oh{n^d}$, and $d$ is constant, the total time needed is $n^{\Oh{ic}}$.
\end{proof}

\subsection{Derandomization using universal hash families}\label{ssct:uhash}
Some of our algorithms use the trick of randomized sampling to obtain a certain structure with good probability and then derandomize this procedure by using a $2$-\emph{universal} family of hash functions. A $2$-\emph{universal hash family} is a family $\mathcal {F}$ of functions from $[n]$ to $[k]$ ($n, k \in \N$ and $k \leq n$) such that for any pair $i$ and $j$ of elements in $[n]$, the number of functions from $\mathcal {F}$ that map $i$ and $j$ to the same element in $[k]$ is at most $|\mathcal {F}| /k$.

The following proposition is a combination of a result of Carter and Wegman~\cite{CW1979JCSS} showing the existence of such families, and the observation that these families can be enumerated in logarithmic space~\cite{Vol1999book}. Later on, we use it to derandomize sampling procedures in some of our algorithms.

\begin{proposition}[Carter and Wegman~\cite{CW1979JCSS}, Proposition 7]\label{prop:cw_uhash}
	Let $n, k \in \N$ with $n \geq k$. One can enumerate a $2$-universal hash family for $\brs{[n] \to [k]}$ in time $n^{\Oh{1}}$ using $\Oh{\log{n}}$ bits of space.	
\end{proposition}

\section{Approximation by layering}\label{sect:layering}
We begin this section with the observation that in a directed graph with maximum out-degree $1$, every connected component contains (as an induced subgraph or otherwise) at most one (undirected) cycle. For such a directed graph $D$, consider the graph $G$ obtained by ignoring arc directions. Because every connected component in $G$ also has at most one cycle, one can find a \emph{minimum} vertex cover for $G$ in polynomial time and logarithmic space using a modified post-order traversal procedure on the connected components. The following lemma formalizes this discussion.

\begin{lemma}\label{lemm:outdeg_1_vc}
	Let $D$ be a digraph on $n$ vertices with maximum out-degree $1$ and let $G$ be the undirected graph obtained by ignoring arc directions in $D$. One can find a minimum vertex cover for $G$ in time $n^{\Oh{1}}$ using $\Oh{\log{n}}$ bits of space.
\end{lemma}
\begin{proof}
	We prove the result via a sequence of claims about the structure of $D$ which enables us to apply \hyperref[algo:TreeVtxCover]{TreeVtxCover} to $G$ and enumerate a minimum vertex cover. 

	\begin{claim}
		Every connected component of $G$ has at most one cycle.
	\end{claim}
	Observe that any path in $G$ corresponds to a directed path in $D$: every vertex of the path except the last has out-degree exactly $1$ in $D$. Similarly, every vertex in a cycle also has out-degree exactly $1$ in $D$. Now consider a connected component in $G$. If it contains two cycles, then the corresponding subgraph of $D$ also contains two directed cycles. They cannot overlap, as this would mean that one of the vertices common to both cycles has out-degree more than $1$ in $D$. In the other case, i.e.\ there is a directed path from a vertex of one cycle to a vertex of the other, the start vertex of this path has out-degree greater than $1$ which is also a contradiction. Thus, the claim is true.

	\begin{claim}
		One can enumerate a minimum vertex cover for every component of $G$ in polynomial time using $\Oh{\log{n}}$ bits of space. 
	\end{claim}
	Consider the following procedure, which finds the cycle (if it exists) in any connected component $C$ of $G$.
	\begin{enumerate}
		\item For each vertex $v \in \V{C}$ with out-degree $1$, set $c \gets 1$ and perform the following steps.
		\item
			\begin{enumerate}
				\item\label{item:tree_vc_cycle} Let $u$ be $v$'s out-neighbour. Set $v \gets u$ and $c \gets c + 1$.
				\item If $c > n$, return $u$.
				\item If $v$ has an out-neighbour, go back to Step~\ref{item:tree_vc_cycle}.
			\end{enumerate}
		\item Return \NO{}.
	\end{enumerate}

	If the above procedure returns \NO{}, then $C$ is cycle-free, i.e.\ $C$ is a tree. In this case, using Algorithm~\ref{algo:TreeVtxCover} on $C$, one can enumerate a minimum vertex cover in polynomial time using $\Oh{\log{n}}$ bits of space.

	In the other case, i.e.\ where the procedure returns a vertex $u$, the component $C$ contains a cycle and $u$ is a vertex in the cycle. Observe that the edge from $u$ to its unique out-neighbour $w$ (in $D$) must be covered by any minimum vertex cover for $C$, i.e.\ either $u$ or $w$ must be in the vertex cover. The graph obtained by deleting either endpoint from $C$ is a tree, since $C$ contains exactly one cycle. We construct two vertex covers, obtained by running Algorithm~\ref{algo:TreeVtxCover} on $C - u$ (resp.\ $C - w$) and augmenting the result with $u$ (resp.\ $w$) to obtain a vertex cover $S_u$ (resp.\ $S_w$) for $C$. One of the two is clearly a minimum vertex cover for $C$, and we enumerate the smaller of the two as the minimum vertex cover.

	The overall process consists merely of running Algorithm~\ref{algo:TreeVtxCover} in a loop where the iteration stops after $\Oh{n}$ steps. Thus, the process takes polynomial time and uses $\Oh{\log{n}}$ bits of space.

	We now prove the main claim, i.e.\ one can enumerate a minimum vertex cover for $G$ in polynomial time using $\Oh{\log{n}}$ bits of space. Observe that combining minimum vertex covers for each component of $G$ produces a minimum vertex cover for all of $G$. Thus, by enumerating the connected components of $G$, one can enumerate a minimum vertex cover for each component in sequence, producing a minimum vertex cover for all of $G$.

	To determine the components and enumerate them, we use the connectivity algorithm of Reingold~\cite{Rei2008JACM}, which runs in polynomial time and uses $\Oh{\log{n}}$ bits of space. This process has a polynomial time overhead and an $\Oh{\log{n}}$-bit space overhead, and thus the entire vertex cover can be enumerated in time $n^{\Oh{1}}$ using $\Oh{\log{n}}$ bits of space as claimed.
\end{proof}

The next result follows directly from Lemma~\ref{lemm:outdeg_1_vc} using the fact that in any graph, the complement of a vertex cover an independent set.

\begin{lemma}\label{lemm:outdeg_1_is}
	Let $D$ be a digraph on $n$ vertices with maximum out-degree $1$ and let $G$ be the undirected graph obtained by ignoring arc directions in $D$. One can find a maximum independent set in $G$ in time $n^{\Oh{1}}$ using $\Oh{\log{n}}$ bits of space.
\end{lemma}

\subsection{$\pVC{}$ on graphs of bounded degree}\label{ssct:bd_vc_2_approx}
We now show that by layering multiple applications of Lemma~\ref{lemm:outdeg_1_vc}, one can compute a $2$-approximate minimum vertex cover in a bounded-degree graph. Our approach is inspired by a distributed algorithm of Polishchuk and Suomela~\cite{PS2009IPL} which computes $3$-approximate solution.

\begin{theorem}\label{thrm:bd_min_vc}
	There is an algorithm which takes as input a graph $G$ on $n$ vertices with maximum degree $\Delta$, and enumerates a $2$-approximate minimum vertex cover for $G$. The algorithm runs in time $n^{\Oh{\Delta}}$ and uses $\Oh{\Delta \log{n}}$ bits of space.
\end{theorem}
\begin{proof}
Set $G_0 = G$ and $V_0 = \V{G}$. The algorithm works in stages $1, \dotsc, \Delta$ as follows. In Stage $i$, it enumerates the subgraph $H_{i - 1}$ of $G_{i - 1}$ in which each vertex of $u$ of $G_{i - 1}$ only retains the edge to its $i$\exS{th} neighbour $v$ (if it exists) in $G$. Observe that directing every such edge from $u$ to $v$ yields a directed graph $R$ with maximum out-degree $1$.

Applying the procedure of Lemma~\ref{lemm:outdeg_1_vc} with $D = R$ and $G = H_{i - 1}$, the algorithm now enumerates a minimum vertex cover $S_i$ for $H_{i - 1}$ in polynomial time using $\Oh{\log{n}}$ bits of space. It then enumerates the graph $G_i$ by removing the vertex set $S_i$ from $G_{i - 1}$ and outputs the vertices in $S_i$. At the end of Stage $\Delta$, the algorithm terminates.

We now prove the bounds in the claim. Observe that the vertex set of $G_i$ ($i \in [\Delta]$) is precisely $\V{G_{i - 1}} \setminus S_i$. In Stage $i$, the algorithm only considers the vertices in $G_{i - 1}$, so the vertex cover generated by it has no neighbours in vertex covers generated in earlier stages, i.e.\ $S_i \cap S_j = \emptyset$ for $j < i$.

For each $H_{i - 1}$, consider a maximal matching $M_i$ in $H_{i - 1}$. From the way the various sets $S_i$ are generated, it is easy to see that $S = \bigcup_{i = 1}^{\Delta} S_i$ forms a vertex cover for $G$ and additionally, $M = \bigcup_{i = 1}^{\Delta} M_i$ is a maximal matching in $G$. Observe that the each set $S_i$ also covers the matching $M_i$ in $H_{i - 1}$. Since $S_i$ is a minimum vertex cover for $H_{i - 1}$, and the endpoints of edges in $M_i$ form a vertex cover for $H_{i - 1}$, we have $\abs{S_i} \leq 2 \abs{M_i}$.

As $M$ is a maximal matching in $G$, the endpoints of edges in $M$ form a vertex cover for $G$, and we have $\abs{S} = \sum_{i = 1}^{\Delta} \abs{S_i} \leq 2 \cdot \sum_{i = 1}^{\Delta} \abs{M_i} \leq 2 \cdot \sum_{i = 1}^{\Delta} \tau(G)$, where $\tau(G)$ is the vertex cover number of $G$. Thus, the set $S$ output by the algorithm is a $2$-approximate vertex cover.

Now observe that for all $i \in [\Delta]$, $G_i$ and $S_i$ satisfy the hypothesis of Theorem~\ref{thrm:graph_oracles}. Thus, one can enumerate each of the sets $S_i$ in time $n^{\Oh{i}}$ using $\Oh{i \log{n}}$ bits of space. Since the maximum value $i$ takes on is $\Delta$, the algorithm runs in time $n^{\Oh{\Delta}}$ and uses a total of $\Oh{\Delta \log{n}}$ bits of space.
\end{proof}

\subsection{$\pdHS{}$ on families with bounded element multiplicity}\label{ssct:bd_dhs_d_approx}
We now consider the $\pdHS{}$ problem, where an instance consists of a finite ground set $U$, a family $\mathcal{F}$ of subsets of $U$ of size at most $d$ and the objective is to compute a hitting set of minimum size for $\mathcal{F}$.

\begin{theorem}\label{thrm:bd_maxl_is}
	There is an algorithm which takes as input a graph $G$ on $n$ vertices with maximum degree $\Delta$ and enumerates a maximal independent set in $G$. The algorithm runs in time $n^{\Oh{\Delta}}$ and uses $\Oh{\Delta \log{n}}$ bits of space.
\end{theorem}
\begin{proof}
	The argument here is essentially the same as that in the proof of Theorem~\ref{thrm:bd_min_vc} with suitable modifications to compute an independent set instead of a vertex cover.

	Set $G_0 = G$ and $V_0 = \V{G}$. The algorithm works in stages $1, \dotsc, \Delta$ as follows. In Stage $i$, it enumerates the subgraph $H_{i - 1}$ of $G_{i - 1}$ where each vertex $u$ of $G_{i - 1}$ only retains the edge to its $i$\exS{th} neighbour $v$ (if it exists) in $G$. Observe that directing every such edge from $u$ to $v$ yields a directed graph $R$ with maximum out-degree $1$.

	Applying the procedure of Lemma~\ref{lemm:outdeg_1_is} with $D = R$ and $G = H_{i - 1}$, the algorithm now enumerates a maximum independent set $S_i$ in $H_{i - 1}$ in polynomial time using $\Oh{\log{n}}$ bits of space. It then enumerates the graph $G_i$ by removing the vertex set $S_i \cup \Nb{S_i}$ from $G_{i - 1}$ and outputs the vertices in $S_i$. At the end of Stage $\Delta$, the algorithm terminates.

	We now show that $S = \bigcup_{i = 1}^{\delta} S_i$ forms a maximal independent set in $G$. Observe that the vertex set of $G_i$ ($i \in [\delta]$) is precisely the set obtained by removing $S_i$ and all vertices incident with $S_i$ from $\V{G_{i - 1}}$. In Stage $i$, the algorithm only considers the vertices in $G_{i - 1}$, whose vertices have no neighbours in $G_j$ for any $j < i - 1$. Thus, $S_i$ has no neighbours in independent sets generated in earlier stages and $S = \bigcup_{i = 1}^{\delta} S_i$ is an independent set in $G$ as well.

	Assume for a contradiction that $S$ is not maximal. Then there is a vertex $v \in \V{G}$ which is not incident with any vertex in $S$. Suppose $v$ is the $i$\exS{th} neighbour of another vertex in $G$. Either $v$ appears in $H_{i - 1}$ or it is excluded because at an earlier stage, some neighbour of $v$ was included in an independent set $S_{j}$ for some subgraph $H_{j - 1}$ with $j < i$. In the latter case, $v$'s neighbour is in $S$, contradicting the assumption. In the former case, if $v$ is not included in the maximum independent set $S_i$ for $H_{i - 1}$, there is a vertex in $S_i$ adjacent to $v$, which is again a contradiction. Thus, $S$ is a maximal independent set in $G$.

	Now observe that for all $i \in [\Delta]$, $G_i$ and $S_i$ satisfy the hypothesis of Theorem~\ref{thrm:graph_oracles}. Thus, one can enumerate each of the sets $S_i$ in time $n^{\Oh{i}}$ using $\Oh{i \log{n}}$ bits of space. Since the maximum value $i$ takes on is $\Delta$, the algorithm runs in time $n^{\Oh{\Delta}}$ and uses a total of $\Oh{\Delta \log{n}}$ bits of space.
\end{proof}

What follows is the main result of this section, which extends Theorem~\ref{thrm:bd_min_vc} to $\pdHS{}$. Observe that for any instance of $\pdHS{}$, the elements of a maximal non-intersecting subfamily of the input forms a $d$-approximate solution. We show that an input family $\mathcal{F}$ of sets can be decomposed into multiple smaller families where it is possible to find maximal non-intersecting subfamilies in polynomial time and logarithmic space using the algorithm of Theorem~\ref{thrm:bd_maxl_is}. These maximal non-intersecting subfamilies can then be combined to obtain a maximal non-intersecting subfamily of $\mathcal{F}$ whose elements form a $d$-approximate solution.

\begin{theorem}
	There is an algorithm which takes as input a ground set $U$ with $n$ elements, a family $\mathcal{F}$ of subsets of $U$ of size at most $d \in \N$ where each element of $U$ appears at most $\delta$ times, and enumerates a $d$-approximate minimum hitting set for $\mathcal{F}$. The algorithm runs in time $n^{\Oh{d{\delta}^2}}$ and uses $\Oh{d{\delta}^2 \log{n}}$ bits of space.
\end{theorem}

\begin{proof}
	Set $\mathcal{F}_0 = \mathcal{F}$ and $U_0 = U$. Let $e \in U$ be an element that appears $\delta_{e}$ times in $\mathcal{F}$. From the ordering of the sets in the input, it is possible to determine (in polynomial time and logarithmic space) the $i$\exS{th} ($i \in [\delta_{e}]$) set in which $e$ appears. We call this the $i$\exS{th} set for $e$. The algorithm works in stages $1, \dotsc, \Delta$ as follows. In Stage $i$, it enumerates the subfamily $\mathcal{H}_{i - 1}$ of $\mathcal{F}_{i - 1}$ which includes all sets $A \in \mathcal{F}_{i - 1}$ such that $A$ is the $i$\exS{th} set for some element in $U_{i - 1}$.

	The algorithm now enumerates a maximal non-intersecting subfamily $\mathcal{K}_i$ of $\mathcal{H}_{i - 1}$. Observe that this can be obtained as a maximal independent set in the intersection graph of $\mathcal{H}_{i - 1}$: each set is represented by a vertex in the intersection graph, and an intersection between any two sets is represented by an edge between the corresponding vertices. One can enumerate this graph by producing $[\mathcal{H}_{i - 1}]$ as the vertex set and producing $ij$ ($1 \leq i < j \leq \abs{\mathcal{H}_{i - 1}}$) as an edge whenever the $i$\exS{th} and $j$\exS{th} sets in $\mathcal{H}_{i - 1}$ intersect. Thus, the graph can be enumerated in polynomial time and logarithmic space. Using Theorem~\ref{thrm:bd_maxl_is} on this graph with $\Delta = d(\delta - 1)$, the algorithm computes a maximal non-intersecting subfamily of $\mathcal{H}_{i - 1}$. This step takes time $n^{\Oh{1}}$ and uses $\Oh{d(\delta - 1) \log{n}}$ bits of space. The algorithm then outputs $S_i = \bigcup \mathcal{K}_i$ and enumerates $U_i = U_{i - 1} \setminus S_i$ as the ground set and $\mathcal{F}_i = \setb{A \in \mathcal{F}_{i - 1}}{A \cap S_i = \emptyset}$ as the subfamily for the next stage. At the end of Stage $\delta$, the algorithm terminates.

	Observe that at each stage, the sets in the maximal subfamily computed do not intersect those in any maximal subfamily computed at later stages. Additionally, each set in $\mathcal{F}$ appears in some subfamily $\mathcal{H}_i$. Using arguments similar to those in the proof of Theorem~\ref{thrm:bd_maxl_is}, one can show that the family $\mathcal{K} = \bigcup_{i = 1}^{\delta} \mathcal{K}_i$ is a maximal non-intersecting subfamily of $\mathcal{F}$. Thus, any hitting set $T$ for $\mathcal{F}$ must contain at least one element from each set in $\mathcal{K}$ and $\abs{T} \geq \abs{\mathcal{K}}$. Because $\mathcal{K}$ is a maximal non-intersecting subfamily of $\mathcal{F}$, any set in $\mathcal{F} \setminus \mathcal{K}$ intersects some set in $\mathcal{K}$, i.e.\ each set in $\mathcal{F}$ contains an element appearing in $\mathcal{K}$. The algorithm outputs $S = \bigcup_{i \in [\delta]} S_i$, the set of elements appearing in $\mathcal{K}$, which is a hitting set for $\mathcal{F}$, and since each set in $\mathcal{K}$ has at most $d$ elements, the size of the set output by the algorithm is at most $d \cdot \abs{\mathcal{K}} \leq d \cdot \abs{T}$. Thus, the set output is a $d$-approximate minimum hitting set for $\mathcal{F}$. 

	We now prove the resource bounds. Observe that for all $i \in [\delta]$, $S_i$, $U_i$ and $\mathcal{F}_i$ satisfy the hypothesis of Theorem~\ref{thrm:hypergraph_oracles} with $c = d (\delta - 1)$. Thus, one can enumerate each of the sets $S_i$ in time $n^{\Oh{i d (\delta - 1)}}$ using $\Oh{i d (\delta - 1) \log{n}}$ bits of space. Since the maximum value $i$ takes on is $\delta$, the algorithm runs in time $n^{\Oh{d{\delta}^2}}$ and uses a total of $\Oh{d{\delta}^2 \log{n}}$ bits of space.
\end{proof}

\section{Staggered Greedy Heuristics}\label{sect:hs_vc_del}
In this section, we consider the $\pdHS{}$ and $\pDS{}$ problems. We show that by combining greedy strategies with certain kernelization rules, one can devise space-efficient approximation algorithms for both problems. Algorithms for $\pdHS{}$ can be used as subroutines in solving various \emph{deletion} problems, where the objective is to delete the minimum possible number of vertices from a graph so that the resulting graph satisfies a certain property. As a corollary, we devise approximation algorithms for such problems as well.

\subsection{$\pdHS{}$}
The algorithm of Proposition~\ref{prop:fk_hs_kernel} can be used to approximate $\pdHS{}$ as we show below.
\begin{corollary}\label{corr:hs_sqrt_approx}
	Let $U$ be a ground set with $n$ elements and $\mathcal{F}$ be a family of subsets of $U$ of size at most $d \in \N$. One can enumerate an $\Oh{d n^{1 - 1 / d}}$-approximate minimum hitting set for $\mathcal{F}$ in time $n^{\Oh{d^2}}$ using $\Oh{d^2 \log{n}}$ bits of space.
\end{corollary}
\begin{proof}
	Consider the following algorithm. Starting at $k = 1$, run the algorithm of Proposition~\ref{prop:fk_hs_kernel} and repeatedly increment the value of $k$ until $k = n^{1 / d}$ or the algorithm returns a solution of size $\Oh{d {(k + 1)}^d}$ (i.e.\ it does not return a \NO{} answer) for the first time. If $k$ is incremented until $n^{1 / d}$, then simply return the entire universe as the solution. Clearly, the approximation ratio is $n^{1 - 1 / d}$, as $OPT \geq n^{1 / d}$ (and so the size of the solution returned is $n = n^{1 - 1/d} \cdot n^{1 / d} \leq n^{1 - 1 / d} \cdot OPT$, where $OPT$ is the size of the minimum hitting set).

	If $k < n^{1 / d}$, then the size of the solution produced is $\Oh{d {(k + 1)}^d}$, and we know that $OPT \geq k$, since the algorithm had returned \NO{} answers until this point. So the size of the solution produced is $\Oh{d {(k + 1)}^d} = \Oh{d (k+1)^{d - 1} \cdot (OPT+1)} = \Oh{ d n^{1 - 1/d} \cdot (OPT+1)}$. Thus, we have an $\Oh{ d n^{1 - 1 / d}}$-approximation. The bounds on running time and space used follow from the fact that the algorithm of Proposition~\ref{prop:fk_hs_kernel} runs in time $n^{\Oh{d^2}}$ and uses $\Oh{d^2 \log{n}}$ bits of space.
\end{proof}

What follows is the key result en route to developing a space-efficient approximation algorithm for $\pdHS{}$.

\begin{lemma}\label{lemm:hs_eps_approx}
	Let $0 < \epsilon \leq 1$. There is an algorithm which takes as input a family $\mathcal{F}$ of $d$-subsets of a ground set $U$ with $n$ elements and $k \in \N$, and either determines correctly that $\mathcal{F}$ has no hitting set of size at most $k$ or enumerates a hitting set of size $\Oh{(d / \epsilon) k^{1 + \epsilon}}$. The algorithm runs in time $n^{\Oh{d^2 + (d / \epsilon)}}$ and uses $\Oh{(d^2 + (d / \epsilon)) \log{n}}$ bits of space. 
\end{lemma}
\begin{proof}
	Let $i = \ceil{(d - 1) / \epsilon}$. The algorithm performs $i$ rounds of computation, each using $\Oh{\log{n}}$ bits of space to determine a set of elements (accessible by oracle) to be removed in the next round, or determine that $\mathcal{F}$ has no hitting set of size at most $k$.

	\begin{enumerate}
		\item Use the algorithm of Proposition~\ref{prop:fk_hs_kernel} to obtain a subfamily $\mathcal{F}' \subseteq \mathcal{F}$ over the ground set $U' \subseteq U$ such that
		\begin{itemize}
		 	\item $\abs{\mathcal{F}'} \leq c {(k + 1)}^d,\ \abs{U'} = cd {(k + 1)}^d$, and
		 	\item there exists a hitting set $S \subseteq U$ of size at most $k$ in $\mathcal{F}$ if and only if there exists a hitting set $S' \subseteq U'$ and $S'$ is a hitting set for $\mathcal{F}'$.
	  \end{itemize}

		\item Set $U_0 = U'$ and $\mathcal{F}_0 = \mathcal{F}'$. For $j = \brc{1, 2, \dotsc, i - 1}$, perform the following steps.
		\begin{itemize}
			\item Determine $S_j$, the set of all elements in $U_{j-1}$ which appear in at least $c {(k + 1)}^{d - 1 - j \epsilon}$ sets in $\mathcal{F}_{j - 1}$.
			\item Let $U_j = U_{j - 1} \setminus S_j$ and $\mathcal{F}_j = \setb{A \in \mathcal{F}_{j - 1}}{A \cap S_j = \emptyset}$. If there are more than $c(k+1)^{d - j \epsilon}$ sets in $\mathcal{F}_j$, then return \NO{}.
		\end{itemize}

		\item Determine $S_i$, the set of all elements in $U_{i - 1}$ which are in some set in $\mathcal{F}_{i - 1}$. Output $S = \bigcup_{j = 1}^i S_j$.
	\end{enumerate}

	We now prove the correctness of the algorithm. In Step 1, the algorithm obtains the ground set $U'$ and the family$\mathcal{F}'$, using the algorithm of Proposition~\ref{prop:fk_hs_kernel}. Let $l \in [i - 1]$ such that the algorithm answers \NO{} in Step 2 for $j = l$, and otherwise let $l = i$ if it never returns a \NO{} answer in Step 2.
	\begin{claim}
		For all $j \in [l]$, $\mathcal{F}_j$ has at most $c {(k + 1)}^{d - j \epsilon}$ sets.
	\end{claim}
	Consider the case when the algorithm does not return a \NO{} answer. Observe that the claim holds for the base case $j = 1$: $\mathcal{F}_0$ has $c {(k + 1)}^d$ sets, and since the algorithm does not return a \NO{} answer, we have $\abs{\mathcal{F}_1} \leq c {(k + 1)}^{d - j \epsilon}$. For induction, observe that whenever $\abs{\mathcal{F}_j} \leq c {(k + 1)}^{d - j \epsilon}$, the algorithm ensures that $\abs{\mathcal{F}_{j + 1}} \leq c (k+1)^{d - (j + 1) \epsilon}$; otherwise, it returns a \NO{} answer.

	Suppose the algorithm returns a \NO{} answer at some value of $j$ in Step 2, then there are more than $c {(k + 1)}^{d - j \epsilon}$ sets in $\mathcal{F}_j$, which have survived the repeated removal of sets from $\mathcal{F}_0$ up to this point, and they cannot be hit by any $k$ of the elements in $U_j$, since each element can hit at most $c {(k + 1)}^{d - 1 - j \epsilon}$ sets in $\mathcal{F}_j$. Thus, the algorithm correctly infers that the input does not have a hitting set of size at most $k$.

	Once the algorithm has reached Step 3, the number of sets in the residual family, $\mathcal{F}_{i - 1}$ is at most ${(k + 1)}^{d - \brn{\ceil{(d - 1) / \epsilon} - 1} \cdot \epsilon} < k^{d - \brn{(d - 1) / \epsilon - 1} \cdot \epsilon} = k^{1 + \epsilon}$. The set $S_i$ of elements in $U_{i - 1}$ that appear in some set in $\mathcal{F}_{i - 1}$ is trivially also a hitting set. Observe that the sets of elements removed in earlier stages, i.e. $S_0, \dotsc, S_{i - 1}$ together hit all sets in $\mathcal{F}$ not appearing in $\mathcal{F}_{i - 1}$. Thus, the set $S = \bigcup_{j = 0}^i S_j$ output by the algorithm is a hitting set for $\mathcal{F}$.

	\begin{claim}
		The set $S$ output by the algorithm has at most $\brn{{(d - 1) / \epsilon} + d} k^{1 + \epsilon}$ elements.
	\end{claim}
	For each $j \in [i - 1]$, the algorithm ensures that $\abs{\mathcal{F}_{j - 1}} \leq c {(k + 1)}^{d - (j - 1) \epsilon}$ (otherwise, it returns a \NO{} answer). Thus, the number of elements which appear in at least $c {(k + 1)}^{d - 1 - j \epsilon}$ sets is at most $\brn{c {(k + 1)}^{d - (j - 1) \epsilon}} / \brn{c {(k + 1)}^{d - 1 - j \epsilon}} = k^{1 + \epsilon}$, i.e.\ $\abs{S_j} \leq k^{1 + \epsilon}$.

	In Step 3, the algorithm ensures that $\abs{\mathcal{F}_{i - 1}} \leq k^{d - (i - 1) \epsilon} \leq k^{1 + \epsilon}$. Each set in $\mathcal{F}_{i - 1}$ edges and each of these edges can span at most $d$ elements. Thus, the number of elements in $U_{i - 1}$ which appear in some set in $\mathcal{F}_{i - 1}$ $dk^{1 + \epsilon}$, i.e.\ $\abs{S_i} \leq dk^{1 + \epsilon}$. Therefore, the total number of elements output by the algorithm in all three phases is $\abs{S} = \sum_{j = 1}^i \abs{S_j} \leq (i - 1) k^{1 + \epsilon} + dk^{1 + \epsilon} \leq \brn{\ceil{(d - 1) / \epsilon} + d} k^{1 + \epsilon}$.

	\begin{claim}
		The algorithm runs in time $n^{\Oh{d^2 + (d / \epsilon)}}$ and uses $\Oh{(d^2 + (d / \epsilon))\log{n}}$ bits of space.
	\end{claim}
	Observe that in Step 1, the family $\mathcal{F}_0$ is obtained using the algorithm of Proposition~\ref{prop:fk_hs_kernel}, which runs in time $n^{\Oh{d^2}}$ and uses $\Oh{d^2 \log{n}}$ bits of space (for any constant $d$). The output of the algorithm can now be used as an oracle for $G_0$.

	In Step 2, each successive family $\mathcal{F}_j$ ($j \in [i - 1]$) is obtained from $\mathcal{F}_{j - 1}$ by deleting sets containing elements which appear in at least $k^{1 - j \epsilon}$ sets (this test can be performed using $\Oh{\log{n}}$ bits of space). Thus, given oracle access to $\mathcal{F}_{j - 1}$, an oracle for $\mathcal{F}_j$ can be provided which runs in polynomial time and uses $\Oh{\log{n}}$ bits of space.

	Step 3 involves writing out all elements in $U_{i - 1}$ that appear in some set in $\mathcal{F}_{i - 1}$, which can also be done in $\Oh{\log{n}}$ bits of space given oracle access to $G_{i - 1}$. Since the number of oracles created in Step 2 is $i - 1$, the various oracles together run in time $n^{\Oh{i}}$ and use $\Oh{i \log{n}} = \Oh{(d / \epsilon) \log{n}}$ bits of space (Theorem~\ref{thrm:graph_oracles}). Combined with the $n^{\Oh{d^2}}$ time and $\Oh{d^2 \log{n}}$ bits of space used by the oracle of Step 1, this gives bounds of $n^{\Oh{d^2 + (d / \epsilon)}}$ on the running time and $\Oh{(d^2 + (d / \epsilon)) \log{n}}$ bits on the total space used by the algorithm.
\end{proof}

The next result follows from the above lemma. 
\begin{theorem}\label{thrm:hs_eps_approx}
	Let $0 < \epsilon \leq 1$. For instances $(U, \mathcal{F})$ of $\pdHS{}$ with $\abs{U} = n$, one can enumerate an $\Oh{(d/\epsilon) n^{\epsilon}}$-approximate minimum hitting set in time $n^{\Oh{d^2 + (d / \epsilon)}}$ using $\Oh{(d^2 + (d/ \epsilon)) \log{n}}$ bits of space.
\end{theorem}
\begin{proof}
Consider the following algorithm. Starting with $k = 1$, iteratively apply the procedure of Lemma~\ref{lemm:hs_eps_approx} and increment $k$'s value until it returns a family of size $\Oh{(d / \epsilon) k^{1+\epsilon}}$ or $k=\ceil{n^{1-\epsilon}}$. When $k=\ceil{n^{1-\epsilon}}$ return the entire universe as the solution. In this case, $OPT \geq n^{1-\epsilon}$, the size of the solution produced is $n$ and $n \leq n^{\epsilon} \cdot OPT$, so we have a factor-$n^\epsilon$ approximation.

In the other case, the algorithm returns a family of size $\Oh{(d / \epsilon) k^{1 + \epsilon}}$ for some $k$. Note that $OPT \geq k$ (as the algorithm returned \NO{} so far), so the solution produced is of size $\Oh{(d / \epsilon) k^{\epsilon} k}$, which is $\Oh{(d / \epsilon) n^{\epsilon} OPT}$, i.e.\ we have a factor-$\Oh{(d/\epsilon) n^{\epsilon}}$ approximation. As we merely reuse the procedure of Lemma~
\ref{lemm:hs_eps_approx}, the overall running time is $n^{\Oh{(d^2 + (d / \epsilon)) }}$ and the amount of space used is $\Oh{(d^2 + (d / \epsilon)) \log{n}}$ bits.
\end{proof}

The above theorem allows us to devise space-efficient approximation algorithms for a number of \emph{graph deletion} problems. Let $\Pi$ be a hereditary class of graphs, i.e.\ a class closed under taking induced subgraphs. Let $\Phi$ be a set of forbidden graphs for $\Pi$ such that a graph $G$ is in $\Pi$ if and only no induced subgraph of $G$ is isomorphic to a graph in $\Phi$. Consider the problem $\prob{Del--}\Pi_{fin}$ (described below), defined for classes $\Pi$ with finite sets $\Phi$ of forbidden graphs.
\begin{quote}
	\begin{description}
		\item[\textbf{Instance}] $G$, a graph 
		\item[\textbf{Solution}] a minimum-size set of vertices whose deletion yields a graph in $\Pi$
	\end{description}
\end{quote}

The next result is a combination of the fact that $\prob{Del--}\Pi$ can be formulated as a certain hitting set problem and the procedure of Theorem~\ref{thrm:hs_eps_approx}.
\begin{lemma}\label{lemm:del_pi}
	Let $\epsilon \leq 1$ be a positive number. On graphs with $n$ vertices, one can enumerate $\Oh{(1/\epsilon) n^{\epsilon}}$-approximate solutions for $\prob{Del--}\Pi_{fin}$ in time $n^{\Oh{1 / \epsilon}}$ using  $\Oh{(1 / \epsilon) \log{n}}$ bits of space.
\end{lemma}
\begin{proof}
Let $\Phi$ be the (finite) set of forbidden subgraphs characterizing $\Pi$, $d$ be the maximum number of vertices in any graph in $\Phi$ and $G$ be the input graph with $n$ vertices. Start by enumerating the following family.
$$ \mathcal{F}_G = \setb{S \subseteq \V{G}}{G[S]\ \text{contains a graph from}\ \Phi}$$
This can be done by running over all subsets of $V(G)$ of size at most $d$, and checking for each subset $S$ whether $G[S]$ is isomorphic to some graph in $\Phi$. Since there are constantly many graphs in $\Phi$, this procedure takes time $\Oh{n^d}$ and uses $\Oh{d \log{n}}$ bits of space. Now using the procedure of Theorem~\ref{thrm:hs_eps_approx}, enumerate an $\Oh{(d / \epsilon) n^{\epsilon}}$-approximate minimum hitting set for $\mathcal{F}_G$.

Observe that any set of vertices is a hitting set for $\mathcal{F}_G$ if and only if it is a deletion set for $G$ (as an instance of $\prob{Del--}\Pi_{fin}$). Thus, the hitting set enumerated is an $\Oh{(d / \epsilon) n^{\epsilon}} = \Oh{(1 / \epsilon) n^{\epsilon}}$-approximate ($d$ is constant) minimum deletion set for $G$. The procedure runs in time $n^{\Oh{d^2 + (d / \epsilon)}} = n^{\Oh{1 / \epsilon}}$ and uses $\Oh{(d^2 + (d / \epsilon))\log{n}} = \Oh{(1 / \epsilon) \log{n}}$ bits of space. Combined with the enumeration procedure, the overall running time is $n^{\Oh{d \cdot (1 / \epsilon)}} = n^{\Oh{1 / \epsilon}}$ and the amount of space used is $\Oh{d \log{n} + (1 / \epsilon) \log{n}} = \Oh{(1 / \epsilon) \log{n}}$ bits.
\end{proof}

The following list defines problems for which we obtain polylogarithmic-space approximation algorithms using the preceding lemma.
\begin{description}
    \item[\pTVD{}]\hfill\\
		\textbf{Instance:} $(G, k)$, where $G$ is a graph and $k \in \N$\\
		\textbf{Question:} Is there a set $S \subseteq \V{G}$ with $\abs{S} \leq k$ such that $G - S$ has no triangles?
    \item[\pTFVS{}]\hfill\\
		\textbf{Instance:} $(D, k)$, where $D$ is a a tournament and $k \in \N$\\
		\textbf{Question:} Is there a set $S \subseteq \V{D}$ with $\abs{S} \leq k$ such that $G - S$ is acyclic?
    \item[\pCVD{}]\hfill\\
		\textbf{Instance:} $(G, k)$, where $G$ is a graph and $k \in \N$\\
		\textbf{Question:} Is there a set $S \subseteq \V{G}$ with $\abs{S} \leq k$ such that $G - S$ is a disjoint union of cliques, i.e.\ a cluster graph?
    \item[\pSVD{}]\hfill\\
		\textbf{Instance:} $(G, k)$, where $G$ is a graph and $k \in \N$\\
		\textbf{Question:} Is there a set $S \subseteq \V{G}$ with $\abs{S} \leq k$ such that $G - S$ can be partitioned into a clique and an independent set, i.e.\ such that $(G - S)$ is a split graph?
    \item[\pThVD{}]\hfill\\
		\textbf{Instance:} $(G, k)$, where $G$ is a graph and $k \in \N$\\
		\textbf{Question:} Is there a set $S \subseteq \V{G}$ with $\abs{S} \leq k$ such that $G - S$ is threshold graph? A threshold graph is one which can be constructed from a single vertex by a sequence of operations that either add an isolated vertex, or add a vertex which dominates all the other vertices.
    \item[\pCoVD{}]\hfill\\
		\textbf{Instance:} $(G, k)$, where $G$ is a graph and $k \in \N$\\
		\textbf{Question:} Is there a set $S \subseteq \V{G}$ with $\abs{S} \leq k$ such that $G - S$ contains no induced paths of length $4$, i.e.\ it is a cograph?
\end{description}

For all the problems appearing above, the target graph classes are known to be characterized by a finite set of forbidden induced subgraphs (see e.g.\ Cygan et al.~\cite{CFK+2015book}) and so the problems can be formulated as $\prob{Del--}\Pi$. By setting $\epsilon$ to a small positive constant or $(1 / \log{n})$, we obtain the following corollary to Lemma~\ref{lemm:del_pi}.
\begin{corollary}\label{corr:phidel2}
	On graphs with $n$ vertices, one can enumerate
	\begin{itemize}
		\item $\Oh{n^{\epsilon}}$-approximate solutions in time $n^{\Oh{1 / \epsilon}} = n^{\Oh{1}}$ using $\Oh{(1 / \epsilon) \log{n}} = \Oh{\log{n}}$ bits of space for any positive constant $\epsilon \leq 1$, and

		\item $\Oh{\log n}$-approximate solutions in time $n^{\Oh{\log{n}}}$ using $\Oh{\log^2 n}$ bits of space
	\end{itemize}
	for the problems $\pVC{}$, $\pTVD{}$,\\$\pThVD{}$, $\pCVD{}$, $\pSVD{}$,\\$\pCoVD{}$ and $\pTFVS{}$.
\end{corollary}

\subsection{$\pDS{}$}\label{ssct:dom_set}
In this section, we describe approximation algorithms for $\pDS{}$ restricted to certain graph classes. A problem instance consists of a graph $G = (V, E)$ and $k \in\N$, and the objective is to determine if there is a \emph{dominating set} of size at most $k$, i.e.\ a set $S \subseteq V$ of at most $k$ vertices such that $S \cup \Nb{S} = V$. 


The first result of this section concerns graphs excluding $C_4$ (a cycle on $4$ vertices) as a subgraph. On such graphs, one can enumerate $\Oh{\sqrt{n}}$-approximations in time $n^{\Oh{1}}$ and $\Oh{\log{n}}$ bits of space using a known kernelization algorithm~\cite{RS2008Algorithmica}.


\subsubsection{\tops{$C_4$}{C4}-Free Graphs}\label{ssct:c4_ds_approx} 
Any vertex $v \in \V{G}$ of degree at least $2k+1$ must be in any dominating set of size at most $k$, as any other vertex (including a neighbour of $v$) can dominate at most $2$ vertices in the neighbourhood (as there will be a $C_4$ otherwise). Using this, we establish the following result.

\begin{lemma}
	There is an algorithm which takes as input a $C_4$-free graph $G$ on $n$ vertices and $k \in \N$, and either determines that $G$ has no dominating set of size at most $k$, or outputs a dominating set of size $\Oh{k^2}$. The algorithm runs in time $n^{\Oh{1}}$ and uses $\Oh{\log{n}}$ bits of space.
\end{lemma}
\begin{proof}
	Consider the following algorithm.
	\begin{enumerate}
	\item Let $S$ be the set of vertices with degree more than $2k$. If $|S|$ is more than $k$, return \NO{}.

	\item The set $S$ dominates all vertices in $\Nb{S}$. If $\abs{V \setminus (S \cup \Nb{S})} > (2k + 1) \cdot (k - \abs{S})$ return \NO{}, as each vertex in $V \setminus S$ can dominate at most $2k+1$ vertices including itself. 

    \item Output $S \cup (V \setminus (S \cup \Nb{S}))$.
	\end{enumerate}

	Recall that any vertex $v \in \V{G}$ of degree at least $2k+1$ must be in any dominating set of size at most $k$. Correctness is now immediate from the the description of the algorithm. When it outputs vertices, it outputs $S$, which has at most $k$ vertices from Step 1, and the number of remaining vertices in Step 2 is $\Oh{k^2}$, so it outputs $\Oh{k^2}$ vertices overall. To see that the space used is $\Oh{\log{n}}$ bits, observe that membership in each of the sets output is determined by predicates that test degrees of vertices individually, and these predicates can by computed in logarithmic space. Thus, by Theorem~\ref{thrm:graph_oracles}, the algorithm uses a total of $\Oh{\log{n}}$ bits of space.
\end{proof}

The proof of the following corollary uses arguments very similar to those in the proof of Theorem~\ref{thrm:hs_eps_approx}, so we omit it.
\begin{corollary}
	There is an algorithm which takes as input a $C_4$-free graph $G$ on $n$ vertices, and enumerates an $\Oh{\sqrt{n}}$-approximate minimum dominating set for $G$. The algorithm runs in polynomial time and uses $\Oh{\log{n}}$ bits of space.
\end{corollary}
%
\subsubsection{Bounded-Degeneracy Graphs}\label{ssct:dgn_ds_approx}
A graph is called $d$-degenerate if there is a vertex of degree at most $d$ in every subgraph of $G$. Examples include planar graphs, which are $5$-degenerate and graphs with maximum degree $d$, which are trivially $d$-degenerate.

There is a generalization of the polynomial kernel for $\pDS{}$ on $C_4$-free graphs (used in Section~\ref{ssct:c4_ds_approx}) to $K_{i, j}$-free graphs for any fixed $i, j \in \N$~\cite{PRS2012TALG} ($K_{i, j}$ is the complete bipartite graph with $i$ vertices in one part and $j$ vertices in the other). The class of $K_{i, j}$-free graphs includes $C_4$-free graphs and for $i \leq j$, $(i+1)$-degenerate graphs. This kernel however, does not seem amenable to modifications that would allow its use in computing approximate solutions using logarithmic or even polylogarithmic space. To design a space-efficient approximation algorithm for $d$-degenerate graphs, we resort instead to the factor-$\Oh{d^2}$ approximation algorithm of Jones et al.~\cite{JLR+2017SIDMA}. We make several adaptations to achieve an $\Oh{\log^2{n}}$ bound on the space used.

Let $G$ be a $d$-degenerate graph on $n$ vertices. As every subgraph of $G$ has a vertex with degree at most $d$, the number of edges in $G$ is at most $dn$. The following lemma is an immediate consequence of this.

\begin{lemma}\label{lemm:subgr_d_deg}
	In any $p$-vertex subgraph of a $d$-degenerate graph, at least $p/2$ vertices are of degree at most $2d$.
\end{lemma}

The following is a description of our algorithm.

\begin{algorithm}
\KwIn{$G = (V, E)$, a $d$-degenerate graph}
\KwOut{$S$, an $\Oh{d^2}$-approximate minimum dominating set for $G$}

	$W, W_h \gets V$\;
	$W_l, Y, B, B_h, B_l \gets \emptyset$\;
	\While(\tcp*[f]{there are vertices in $W_h$ to be dominated}){$W_h \neq \emptyset$}{
		$W^* \gets W \cup B_h$\;
		$S \gets \setb{v \in W_h}{\degr{G[W^*]}{v} \leq 2d}$\;	
		$Y \gets Y \cup \Nbr{G[W^*]}{S}$\tcp*[f]{$ Y $ is the partial solution}\;
		$B \gets \Nb{Y}$\;
		$W \gets V \setminus (Y \cup B)$\;\label{algo:DgnDomSet.Wu}
		$B_h \gets \setb{v \in B}{\degr{G[W]}{v} \geq 2d + 1}$\;
		$B_l \gets B \setminus B_h$\;
		$W^* \gets W \cup B_h$\;
		$W_h \gets \setb{v \in W}{v\ \text{is not isolated in}\ G[W^*]}$\;
		$W_l \gets W \setminus W_h$\;
	}
\Return{$Y \cup W_l$}

\caption{DgnDomSet: find an approximate minimum dominating set}\label{algo:DgnDomSet}
\end{algorithm}

The algorithm starts by picking the neighbours of all vertices (they form the set $S$) of degree at most $2d$, and repeatedly finds such vertices in smaller and smaller subgraphs of $G$, picking all their neighbours into the solution as well. As each such vertex or one of its neighbours must be in any dominating set, this will result in an $\Oh{d}$-approximate solution if we manage to find a vertex that dominates (at least one and) at most $2d$ of the non-dominated vertices (set $W$ on Line~\ref{algo:DgnDomSet.Wu}). This may not happen in the intermediate steps as more and more vertices are dominated by those vertices picked earlier. So the algorithm carefully partitions the set of undominated vertices.

\begin{figure}
\begin{center}
\includegraphics[scale=1]{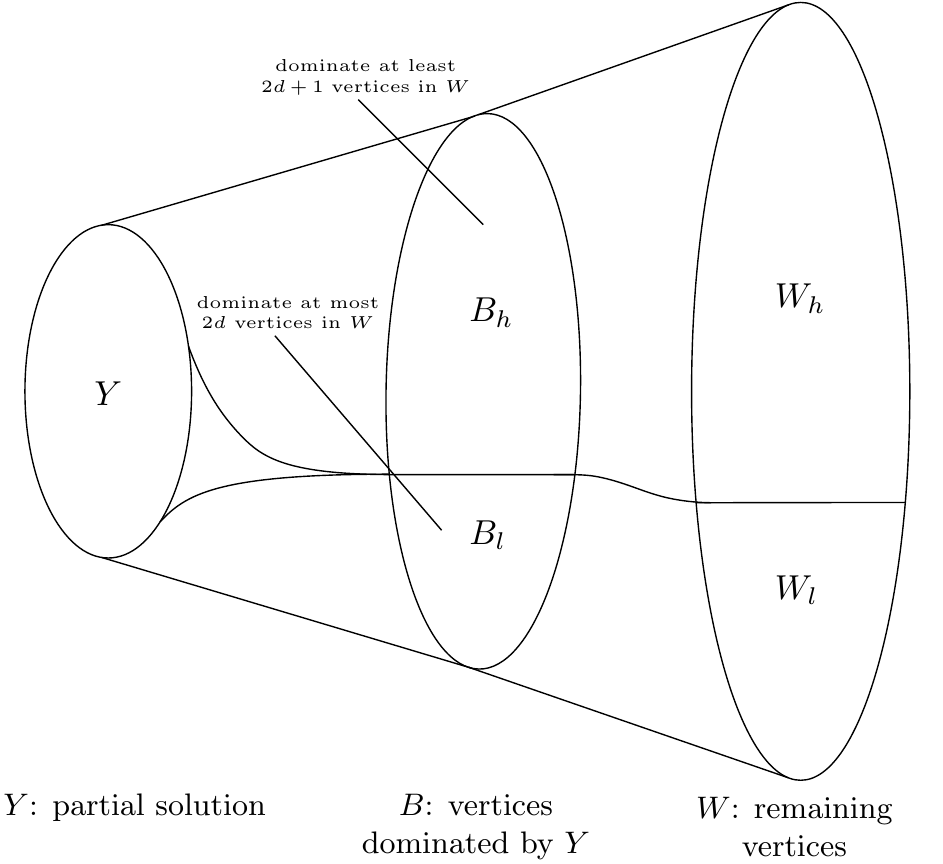}
\end{center}
\caption{Partitioning of the vertices in the algorithm for $\pDS{}$ on $d$-degenerate graphs.}\label{fig:dgn_ds_partition}
\end{figure}

Let $Y$ be the set of vertices picked at any point, $B$ be the set of vertices (other than those in $Y$) dominated by $Y$, and $W$ be the set of vertices in $V \setminus (Y \cup B)$ (see Figure~\ref{fig:dgn_ds_partition}). The goal is to dominate vertices in $W$, and we try to do so by finding (the neighbours of) low degree vertices from $B \cup W$. So we start finding low degree (at most $2d$) vertices in $B \cup W$ to pick their neighbours. First we look for such vertices in $B$, and so we further partition $B$ into $B_h$, those vertices of $B$ with at least $2d+1$ neighbours in $W$ and $B_l = B \setminus B_h$.


First, we remove (for later consideration) vertices of $W$ that have no neighbours in $W \cup B_h$, let they be $W_l$  and focus on the induced subgraph $G[B_h \cup W_h]$ where $W_h = W \setminus W_l$. Here, we are bound to find low degree vertices from $W_h$ (as vertices in $B_h$ have high degree) as long as $W_h$ is non-empty, and so we repeat the above procedure of picking the neighbours of all low degree vertices from $W_h$. Eventually, when $W_h$ becomes empty, if $W_l$ is non-empty, we simply pick all vertices of $W_l$ into the solution. What follows is a pseudocode description of the algorithm.



If we treat a round as the step where we find all vertices in $W_h$ with at most $2d$ neighbours in $W_h$, then as at least a fraction of the vertices of $W_h$ are dominated in each round (Lemma~\ref{lemm:subgr_d_deg}), the number of rounds is $\Oh{\log n}$. Each round just requires identifying vertices based on their degrees in the resulting subgraph, the $i$-th round can be implemented in 
$\Oh{i \log n}$ bits using Theorem~\ref{thrm:graph_oracles} resulting in an $\Oh{\log^2 n}$ bits implementation.


The approximation ratio of $\Oh{d^2} $ can be proved formally using a charging argument (see Jones et al.~\cite{JLR+2017SIDMA}, Theorem 4.9). We give an informal explanation here. First we argue the approximation ratio of $(2d+1)$ for the base case when $W_h$ is empty. Isolated vertices in $W_l$ are isolated vertices in $G$ and hence they need to be picked in the solution.
The number of non-isolated vertices in $W_l$ is at most $2d |B_l|$ as their neighbours are only in $B_l$ (otherwise, by definition, those vertices will be in $W_h$). As vertices in $B_l$ have degree at most $2d$, $|W_l| \leq 2d |B_l|$ and as at least one vertex of $B_l \cup W_l$ must be picked to dominate a vertex in $W_l$, we have the approximation ratio of $(2d+1)$ for those vertices.

In the intermediate step, if we did not ignore vertices in $B_l$ to dominate a vertex in $W_h$, a $(2d+1)$- approximation is clear. For, a vertex or one of its at most $2d$ neighbours must be picked in the dominating set. However, a
vertex in $W_h$ maybe dominated by a vertex in $B_l$, but by ignoring $B_l$, we maybe picking $2d$ vertices to dominate it. As a vertex in $B_l$ can dominate at most $2d$ such vertices of $W_h$, we get an approximation ratio of $\Oh {d^2} $. 

The next theorem formalizes the above discussion.
\begin{theorem}
	There is an algorithm which takes as input a $d$-degenerate graph on $n$ vertices and enumerates an $\Oh{d^2}$-approximate minimum dominating for it. The algorithm runs in time $n^{\Oh{\log{n}}}$ and uses $\Oh{\log^2{n}}$ bits of space.
\end{theorem}

\section{Randomization}\label{sect:randomization}
In this section, we devise approximation algorithms for restricted versions of $\pIS{}$ and $\pDS{}$ using hash families constructible in logarithmic space to derandomize known randomized sampling procedures.

\subsection{$\pIS{}$ on graphs with bounded average degree}\label{ssct:avd_is_approx}
On general graphs, the problem is unlikely to have a non-trivial (factor-$(n^{1 - \epsilon})$) approximation algorithm~\cite{Has1999ActaMath}. However, if the graph has average degree $d$, then an independent set satisfying the bound of the next lemma is a $(2d)$-approximate solution.
Note that graphs of bounded average degree encompass planar graphs and graphs of bounded degeneracy. It is also known that $2d$ is the best approximation ratio possible up to polylogarithmic factors in $d$~\cite{AKS2011TOC,Cha2016JACM}.

\begin{proposition}[Alon and Spencer~\cite{AS2008book}, Theorem 3.2.1]\label{prop:avd_is}
	If a graph on $n$ vertices has average degree $d$, then it has an independent set of size at least $n / (2d)$.
\end{proposition}

In what follows, we develop a logarithmic-space procedure that achieves the above bound. Let $G = (V, E)$ be a graph on $n$ vertices with average degree $d$. Consider a set $S \subseteq V$ obtained by picking each vertex in $V$ independently with probability $p =  1 / d$. Let $m_S$ be the number of edges with both endpoints in $S$. The following bound appears as an intermediate claim in the proof of Proposition~\ref{prop:avd_is} (see Alon and Spencer~\cite{AS2008book}, Theorem 3.2.1). We use it here without proof.

\begin{lemma}\label{lemm:avd_is_bound}
	$\Ex{\abs{S} - m_S} = n / (2d)$.
\end{lemma}

Consider the set $I$ obtained by arbitrarily eliminating an endpoint of each edge in $G[S]$. Observe that $G[I]$ has no edges, i.e.\ $I$ is an independent set whose expected size is $\Ex{\abs{S} - m_S} = n / (2d)$.

Derandomizing this sampling procedure is simple: simply run through the functions of a $2$-universal hash family $\mathcal{F}$ for $\brs{[n] \to [d]}$ and for each $f \in \mathcal{F}$, pick a vertex $v \in V$ into $S$ if and only if $f(v) = 1$. Because the range of the functions is $[d]$, the sampling probability is $\Pr{v \in S} = 1 / d$. Recall that Lemma~\ref{lemm:avd_is_bound} only requires the sampling procedure to be pairwise independent, so the expectation bound remains the same: $\Ex{\abs{S} - m_S} = n / (2d)$. While going through $\mathcal{F}$, select the function $f \in \mathcal{F}$ which maximizes $|S| - m_S$, where $S = \setb{v \in V}{f(v) = 1} $ and $m_S$ is the number of edges $uv \in E$ with $f(u) = f(v) = 1$. 
Using the construction of Proposition~\ref{prop:cw_uhash}, this step can be performed in polynomial time using $\Oh{\log{n}}$ bits of space and $f$ can be used as an oracle for $S$ at the same space cost.

The next step, in which vertices are deleted arbitrarily from each pair of adjacent vertices in the sample $S$, is tricky to carry out in small space. This is because for any edge $uv$ in $G[S]$, it is not possible to determine whether either of the endpoints survive the deletion procedure without additional information about the other edges incident with $u$ and $v$. However, we can achieve this by using the ordering induced on the vertices by the input encoding to ensure that vertices in $S$ are retained only if they are the smallest (in the input ordering) vertices in their neighbourhoods in $G[S]$. Using this, we prove the following lemma.

\begin{lemma}\label{lemm:is_small_nbd}
	Let $T$ be the set of vertices $v \in S$ such that $v$ is the smallest vertex in its neighbourhood in $G[S]$. The set $T$ is independent in $G$, has size $\abs{T} \geq \abs{S} - m_S$, and one can enumerate $T$ in polynomial time using $\Oh{\log n}$ bits of space.
\end{lemma}
\begin{proof}
	Determining if $v \in S$ is the smallest vertex in its neighbourhood in $G[S]$ involves enumerating the neighbourhood of $v$ in the induced subgraph $G[S]$ which can be performed in polynomial time using $\Oh{\log{n}}$ bits of additional space. As we pick only one vertex from each neighbourhood, the picked set $T$ is independent and it is trivial to see that the overall procedure is polynomial-time and uses $\Oh{\log{n}}$ bits of space.

	Let $C_1, \dotsc, C_t$ be the connected components of $G[S]$. Consider the difference between the number of vertices and the number of edges in each component. Any component with $l$ vertices contains at least $l - 1$ edges. For $i \in [t]$, denote by $n_i$ the number of vertices in $C_i$ and by $m_i$, the number of edges. We have $\sum_{i = 1}^t (n_i - 1) \leq \sum_{i = 1}^t m_i$, i.e.\ $\sum_{i = 1}^t n_i - t \leq \sum_{i = 1}^t m_i = m_S$, which implies that $t \geq n - m_S$. As we pick at least one vertex (the smallest vertex) from each component in $T$, we have $\abs{T} \geq t \geq n - m_S$.
\end{proof}

We now have the following theorem as a direct consequence of the above results.

\begin{theorem}\label{thrm:is_2d_approx}
	There is a an algorithm which takes as input a graph $G$ on $n$ vertices with average degree $d$, and enumerates a $(2d)$-approximate maximum independent set in $G$. The algorithm runs in time $n^{\Oh{1}}$ and uses $\Oh{\log{n}}$ bits of space.
\end{theorem}

\subsection{$\pDS{}$ on $d$-regular graphs}\label{ssct:reg_ds_approx}
In what follows, we use similar techniques as above to devise a factor-$(\log{(d + 1)} + 1)$ approximation algorithm for $\pDS{}$ restricted to $d$-regular graphs.

\begin{proposition}[Alon and Spencer~\cite{AS2008book}, Theorem 1.2.2] 
	Any graph on $n$ vertices with minimum degree $d$ has a dominating set of size at most $n (\log{(d + 1)} + 1) / (d + 1)$.
\end{proposition}

On a $d$-regular graph, because the size of any dominating set is at least $n / (d + 1)$, the approximation ratio achieved is $\log{(d + 1)} + 1$.

Now we outline the proof of the above proposition to show how it can be derandomized. Consider a $d$-regular graph $G$ on $n$ vertices. Picking each vertex of $G$ with probability $p = \log{(d + 1)} / (d + 1)$ yields a set $S$ with expected size $\Ex{\abs{S}} = np$. By adding in the vertices not dominated by $S$, we obtain a dominating set $W =  S \cup (V \setminus (S \cup \Nb{S}))$. The expected size of this set is $\Ex{\abs{W}} \leq n(p + {(1 - p)}^{d + 1})$, and it can be shown that this quantity is $n(\log{(d + 1)} + 1) / (d + 1)$.

Note that the expectation bounds only need the sampling of the vertices to be pairwise independent. Consider a $2$-universal hash family $\mathcal{F}$ for $\brs{[n] \to [d + 1]}$, and define $S_f = \setb{v \in \V{G}}{f(v) \leq \log{(d + 1)} + 1}$ and $W_f = S_f \cup (V \setminus (S_f \cup \Nb{S_f}))$. Over functions $f = \mathcal{F}$, the sampling probability $\Pr{v \in S_f}$ is $\floor{(\log{(d + 1)} + 1) / (d + 1)}$. Because $\mathcal{F}$ is a $2$-universal hash family, there is a function $f \in \mathcal{F}$ for which $W_f$ achieves the expectation bound for $\abs{W}$ above.

The sampling procedure can now be derandomized as follows. Begin by enumerating $\mathcal{F}$ in logarithmic space using Proposition~\ref{prop:cw_uhash}. For each $f \in \mathcal{F}$, determine $\abs{W_f}$, and output $W_f$ for the first function $f$ for which $\abs{W_f} \geq n(\log{(d + 1)} + 1) / (d + 1)$.

We thus have the following result.
\begin{theorem} 
	There is an algorithm which takes as input a $d$-regular graph $G$ on $n$ vertices and enumerates a $(\log{(d + 1)} + 1)$-approximate minimum dominating set for $G$. The algorithm runs in time $n^{\Oh{1}}$ and uses $\Oh{\log{n}}$ bits of space.
\end{theorem}


%
%
%
\section{Conclusion}
We devised space efficient approximation algorithms for $\pdHS{}$ (and its restriction $\pVC{}$), $\pIS{}$ and $\pDS{}$ in some special classes of graphs.

The algorithms all require random access to their inputs. It is possible to translate them for the streaming model: each time the algorithm reads an element from the input, it makes a single pass. In this way, the space bounds for the translated algorithms remain essentially the same (additive overhead of $\Oh{\log{n}}$ bits for reading in passes), but they require as many passes over their inputs as the running times of the original algorithms. 

We consider our contribution as simply drawing attention to a direction in the study of approximation algorithms, and believe that it should be possible to improve the approximation ratios and the space used for the problems considered here.  Obtaining a constant-factor or even factor-$\Oh{\log{n}}$ approximation algorithm for $\pVC{}$ and a factor-$\Oh{\log n}$ approximation algorithm for $\pDS{}$ on general graphs using $\Oh{\log{n}}$ bits of space are some specific open problems of interest.


%
%

\bibliography{external/references}
\end{document}